\documentclass[10pt]{article}
\usepackage{dcolumn}
\usepackage{bm}
\usepackage{verbatim}       

\usepackage{graphicx}

\usepackage{amsthm}
\usepackage{amssymb}

\usepackage{bm}
\usepackage{amstext}
\usepackage{psfrag}
\usepackage{amsmath}
\usepackage{amsfonts}
\usepackage{units}
\usepackage{mathrsfs}
\usepackage{color}
\newcommand{\p}{\partial}

\newcommand{\bd}{\begin{definition}}                
\newcommand{\ed}{\end{definition}}                  
\newcommand{\bc}{\begin{corollary}}                 
\newcommand{\ec}{\end{corollary}}                   
\newcommand{\bl}{\begin{lemma}}                     
\newcommand{\el}{\end{lemma}}                       
\newcommand{\bp}{\begin{proposition}}            
\newcommand{\ep}{\end{proposition}}                
\newcommand{\bere}{\begin{remark}}                  
\newcommand{\ere}{\end{remark}}                     

\newcommand{\bt}{\begin{theorem}}
\newcommand{\et}{\end{theorem}}

\newcommand{\be}{\begin{equation}}
\newcommand{\ee}{\end{equation}}

\newcommand{\bit}{\begin{itemize}}
\newcommand{\eit}{\end{itemize}}
\newtheorem{theorem}{Theorem}[section]
\newtheorem{corollary}[theorem]{Corollary}
\newtheorem{lemma}[theorem]{Lemma}
\newtheorem{proposition}[theorem]{Proposition}
\theoremstyle{definition}
\newtheorem{definition}[theorem]{Definition}
\theoremstyle{remark}
\newtheorem{remark}[theorem]{Remark}
\newtheorem{example}[theorem]{Example}

\usepackage{authblk}
\usepackage[backref]{hyperref}

\begin{document}

\title{Causal simplicity and (maximal) null pseudoconvexity }

\author[1]{J. Hedicke\thanks{E-mail: jakob.hedicke@ruhr-uni-bochum.de}}
\author[2]{E. Minguzzi\thanks{E-mail: ettore.minguzzi@unifi.it}}
\author[3]{B. Schinnerl\thanks{E-mail: benedict.schinnerl@univie.ac.at}}
\author[3]{R. Steinbauer\thanks{E-mail: roland.steinbauer@univie.ac.at}}
\author[1]{S. Suhr\thanks{E-mail: Stefan.Suhr@ruhr-uni-bochum.de}}

\affil[1]{Fakult\"at f\"ur Mathematik, Ruhr-Universit\"at Bochum, Universit\"atstra\ss e 150, D-44801 Bochum, Germany.}
\affil[2]{Dipartimento di Matematica e Informatica ``U. Dini'', Universit\`a
degli Studi di Firenze, Via S. Marta 3,  I-50139 Firenze, Italy.}
\affil[3]{Fakult\"at f\"ur Mathematik, Universit\"at Wien, Oskar-Morgenstern-Platz 1, A-1090 Wien, Austria.}

\date{}

\maketitle

\begin{abstract}
\noindent
We consider pseudoconvexity properties in Lorentzian and Riemannian manifolds and
their
relationship in static spacetimes. We provide an example of a
causally continuous
and maximal null pseudoconvex spacetime that fails to be causally simple. Its Riemannian factor provides an analogous example of a manifold that is minimally pseudoconvex, but fails to be convex.
\end{abstract}

\section{Introduction}
A pseudo-Riemannian manifold is said to have a pseudoconvex class of geodesics, if for each compact subset $K$ there is a larger compact subset $K'$, such that any geodesic segment of the
class with endpoints in $K$ lies entirely in $K'$. Causal pseudoconvexity is a kind of an ``internal completeness'' assumption for spacetimes akin to, but strictly weaker than global hyperbolicity. Similar to the latter, it is a causality notion with strong ties to the theory of PDE. Pseudoconvexity of bicharacteristics characterizes the existence of a parametrix for pseudodifferential operators of real principal type, cf.\ \cite[Section 6]{DH:72}, \cite{BP:83}.

John Beem and his coauthors have used pseudoconvexity in several contexts in causality theory. In \cite{BE:87} causal pseudoconvexity together with causal geodesic non-imprisonment (inextendible causal geodesics leave every compact set) appear as sufficient conditions for the stability of causal geodesic completeness, cf.\ \cite[Th.\ 7.35]{BEE:96}. Geodesic non-imprisonment and pseudoconvexity of all geodesics together with the absence of conjugate points implies geodesic connectedness and serve as conditions for a pseudo-Riemannian version of the Hadamard-Cartan theorem \cite{BP:89}, \cite[Ch.\ 11]{BEE:96}. For a recent generalization see \cite{CF:21}.

Here we are especially interested in the relation between pseudoconvex and causally simple spacetimes, i.e.\ causal spacetimes with closed causal relation. Indeed, every causally simple spacetime is maximal null pseudoconvex by \cite[Th.\ 1]{beem92}, for which we provide a simplified proof in Theorem \ref{thm:beem92}, below. Concerning the reverse implication it has been conjectured in \cite[Sec.\ 1]{BILL:17} that strongly causal (i.e.\ there are no ``almost closed'' causal curves) and null pseudoconvex spacetimes are causally simple. Finally, in \cite[Th.\ 2]{vatandoost19} it was claimed that for strongly causal spacetimes maximal null pseudoconvexity and causal simplicity are equivalent. Here we provide a counterexample to this statement, i.e.\ we establish that
\begin{center}
 causal continuity
 and maximal null pseudoconvexity \\ do not imply causal simplicity.
\end{center}

Our counterexample is a static spacetime based on a corresponding Riemannian counterexample, that enjoys a certain limiting property for minimizing geodesics, but fails to be convex.
This counterexample adds to the list of recently found counterexamples involving the notion of causal simplicity \cite{hedicke19,minguzzi20a}.

In Section  \ref{sec:caus} we provide some general results on pseudoconvexity and in Section \ref{sec:stat} we specialize to static spacetimes. This will allow us
 to ``lift'' the corresponding properties of the Riemannian counterexample to the spacetime level in Section \ref{sec:r-counter}.
\medskip

In the remainder of this introduction we fix some notations and conventions.
All manifolds are assumed to be smooth, connected, Hausdorff, second countable, of arbitrary dimension $n\geq2$, and without boundary. By a {\em minimizing geodesic} in a Riemannian manifold $(\Sigma, h)$ we mean a geodesic whose length equals the distance $d^h(x,y)$ between its endpoints $x$ and $y$. We call $\Sigma$ {\em convex} if any pair of its points can be connected via a minimizing geodesic.

A spacetime $(M,g)$ is a time oriented Lorentzian manifold, where we use the signature $(-,+,\dots,+)$. A causal geodesic in $M$ is called {\em maximizing} if its length equals the Lorentzian distance $d^g(p,q)$ between its endpoints $p$ and $q$. We denote the chronological and the causal relation by $I$ and $J$ respectively. A spacetime is called {\em causal} if there are no closed causal curves. If in addition the causal relation $J$ is closed, it is called {\em causally simple}.  A spacetime is {\em non-total imprisoning} if no inextendible causal curve is contained in a compact set. We shall only consider causal spacetimes, in which case the non-total imprisonment property is equivalent with the causal geodesic non-imprisonment property mentioned above \cite[Prop.\ 4.41]{minguzzi18b}.
A spacetime is called {\em strongly causal} if for every point  and for every neighborhood of the point there is a smaller
 neighborhood such that no causal curve intersects it more than once. A spacetime is called {\em causally continuous} if it is strongly causal and {\em reflective}: $I^+(q) \subset I^+(p) \Leftrightarrow I^-(p)\subset I^-(q)$. It is known that causal simplicity $\Rightarrow$ causal continuity $\Rightarrow$ strong causality $\Rightarrow$ non-total imprisonment $\Rightarrow$ causality. For other results and conventions on causality not explicitly mentioned in this work, we refer the reader to the review \cite{minguzzi18b}.

\section{General results on pseudoconvexity}\label{sec:caus}

 We start recalling some definitions, cf.\ \cite[Chs.\ 7, 11]{BEE:96}.

\begin{definition}
A spacetime $(M,g)$ is called (causal, null or maximally null) pseudoconvex, if for any compact set $K$, there exists another compact set $K'$, such that each geodesic of the respective type with both endpoints in $K$ must be entirely contained in $K'$.
\end{definition}

Clearly pseudoconvexity implies causal pseudoconvexity which implies null pseudoconvexity which again is stronger than maximal null pseudoconvexity.
There is also a Riemannian version of the notion:

\begin{definition}
A Riemannian manifold $(\Sigma,h)$ is called (minimally) pseudoconvex, if for any compact set $C$, there exists another compact set $C'$, such that each (minimal) geodesic with endpoints in $C$ must be entirely contained in $C'$.
\end{definition}

Clearly, pseudoconvexity implies minimal pseudoconvexity.
Often it is useful to relate pseudoconvexity to a certain limiting property of geodesic segments which we define next.

\begin{definition}
We say a pseudo-Riemannian manifold $M$ has the limit geodesic segment property (LGS) if the following holds true: Given any pair of converging sequences of points $p_n\to p$ and $q_n\to q\not=p$ and any sequence $\sigma_n$ of geodesic segments connecting $p_n$ to $q_n$, there is a subsequence of $\sigma_n$ (in a suitable affine reparametrization) converging (locally uniformly) to a geodesic $\sigma$ from $p$ to $q$.
\end{definition}

In the Riemannian case we will speak of the minimal LGS  if all $\sigma_n$ (and hence $\sigma$) have the corresponding property. Similarly, in the Lorentzian case,
we will speak of the causal, null, maximal null LGS, if all $\sigma _n$ (and hence $\sigma$) have the corresponding property.

Recall from  \cite[Def.\ 2]{BP:89} that a Riemannian manifold $(\Sigma,h)$ is {\it disprisoning} if no forward inextensible geodesic $\gamma\colon [0,\omega)\to
\Sigma$ has compact closure.

\begin{lemma} \label{viw}
A Riemannian manifold $(\Sigma,h)$ satisfies the LGS if and only if it is disprisoning and pseudoconvex.
\end{lemma}

\begin{proof}
First assume that LGS holds. Let $C\subset \Sigma$ compact be given. Choose a compact exhaustion $\{C_n\}_{n\in\mathbb{N}}$ of $\Sigma$.
If pseudoconvexity fails for $C$ there exists a sequence $\gamma_n\colon[0,1]\to \Sigma$ with endpoints in $C$ such that $\gamma_n$ leaves $C_n$.
By the LGS $\{\gamma_n\}_{n\in\mathbb{N}}$ contains a convergent subsequence. The limit geodesic $\gamma\colon [0,1]\to \Sigma$ is contained in
some $C_N$. This clearly contradicts the initial assumption. Therefore $(\Sigma,h)$ is pseudoconvex.

Next assume that LGS holds and there exists a inextensible geodesic ray $\gamma\colon[0,\omega)\to \Sigma$ with compact closure.
Then we have $\omega=\infty$ and $\gamma$ has infinite length. In that case the sequence  $\gamma_n\colon[0,1]\to\Sigma$ where $\gamma_n(t):=
\gamma(nt)$ has no convergent subsequence.

Finally, assume that $(\Sigma,h)$ is pseudoconvex and disprisoning. Let $p_n,q_n\in\Sigma$ be sequences converging to $p$ and $q\neq p$ respectively
and $\gamma_n\colon [0,1]\to \Sigma$ be a sequence of geodesics from $p_n$ to $q_n$.  Let $C$ be a compact set given by the union of a compact neighborhood of $p$ with a compact neighborhood of $q$. Without loss of generality  we can assume $p_n,q_n\in C$ for each $n$.
There exists a compact set $C'\subset \Sigma$ such that
$\gamma_n\subset C'$ for all $n\in\mathbb{N}$ by pseudoconvexity. Since the manifold is disprisoning the length $L^h(\gamma_n)$ is uniformly
bounded from above. Then by a standard argument for geodesics it follows that there exists a convergent subsequence.
\end{proof}

By the previous result every compact (closed) Riemannian manifold is trivially pseudoconvex but fails to satisfy the LGS as, certainly, it  is not disprisoning. For instance,  $T^2$ is pseudoconvex but does not satisfy the LGS. Choosing a geodesic $\sigma$ with irrational slope, which hence is dense in $T^2$, we may choose a point $x\not\in\sigma$ and $t_n$ such that $\sigma(t_n)\to x$. But the sequence of geodesic segments $\sigma|_{[0,t_n]}$ has no subsequence converging to a geodesic between $\sigma(0)$ and $x$.

In this connection we mention a  Theorem of Serre, see \cite{serre51}, implying that any pair of points in a noncontractible complete Riemannian manifold are connected by a sequence of geodesics whose lengths are diverging.

Completeness implies convexity by the   Hopf-Rinow theorem \cite{jost11}, but completeness (hence convexity)  does not imply pseudoconvexity (e.g.\ a complete surface with infinitely many holes), cf.\  \cite{sanchez01}. For another example of convex but not pseudoconvex space see Theorem 2.7 in \cite{hedicke19}.

The next result is analogous to the previous one, but for the ``minimal'' case.

\begin{lemma}\label{lem:mpc}
A Riemannian manifold $(\Sigma,h)$ satisfies the minimal LGS if and only if it is minimally pseudoconvex.
\end{lemma}

\begin{proof}
We already know that if LGS holds then pseudoconvexity holds, which implies minimal pseudoconvexity.

For the converse, let $\Sigma$ be minimally pseudoconvex. Further let $\sigma_n:[0,b_n]\to \Sigma$ be minimizing geodesics from $x_n$ to $y_n$, parametrized by arc length, and let $x_n\to x$, $y_n\to y\not=x$. Let $C$ be a compact set given by the union of a compact neighborhood of $x$ with a compact neighborhood of $y$. Without loss of generality  we can assume $x_n,y_n\in C$ for each $n$.
By compactness $\sigma_n'(0)$ converges up to a subsequence to a unit vector $v\in T_x\Sigma$.
Let $\sigma:[0,\beta)\to\Sigma$ be the maximally extended geodesic with initial vector $v$. Then $\sigma_n$ converges to $\sigma$ locally uniformly on $[0,\beta)$. By minimal pseudoconvexity all $\sigma_n$ are contained in some compact $C'$ and so $L^h(\sigma_n) \leq\text{diam}^h(C')$. Hence, again up to a subsequence, $b_n$ converge to some $b<\beta$. Hence $\sigma_n(b_n)\to \sigma(b)=q$ and we are done.
\end{proof}

The next result is a Lorentzian analog to \ref{viw}. It slightly improves \cite[Lem.\ 11.20]{BEE:96}.

\begin{lemma} \label{lem:mnpc}
$(M,g)$ satisfies the (maximal) null LGS  iff it is  non-total imprisoning and   (maximal) null pseudoconvex.
\end{lemma}

A similar results with ``causal'' replacing the two instances of ``null'' holds.

\begin{proof}
The proof of ``(maximal) null LGS $\Rightarrow$ (maximal) null pseudoconvex'', can be done in the same way as the proof of Lemma \ref{viw} (first paragraph). If non-total imprisonment were violated we could find a future inextendible lightlike geodesic $\gamma:[0,a)\to M$, achronal hence maximizing, contained in a compact set $C$ cf.\ \cite[Thm.\ 2.77]{minguzzi18b}. Let $a_n\to a$, and pass to a subsequence (denoted in the same way) so that $\gamma(a_n)\to q$. But then the maximizing lightlike geodesics $\gamma_n= \gamma \vert_{[0, a_n]}$ with $a_n\to a$, do not converge to a geodesic $\eta$ connecting $p:=\gamma(0)$ to $q$, for if that were the case, as each $\gamma_n$ coincides with a segment of $\gamma$, $\eta$ would coincide with a segment of $\gamma$, which is impossible since certainly $\gamma_n$ converges on a domain larger than any subinterval $[0,b]\subset [0,a)$.

For the converse, assume non-total imprisonment and (maximal) null pseudoconvexity. Let $\sigma_n$ be (maximal) null geodesics with endpoints $p_n$, $q_n$ converging to $p$ and $q$, respectively.  Let $C$ be a compact set given by the union of a compact neighborhood of $p$ with a compact neighborhood of $q$. Without loss of generality  we can assume $p_n,q_n\in C$ for each $n$. By (maximal) null pseudoconvexity there is a compact set $C'$ containing all  the curves $\sigma_n$. By the limit curve theorem (two endpoints case \cite{minguzzi18b}), the sequence must admit a converging subsequence, otherwise we could find a future inextendible causal curve $\sigma^p \subset C'$ starting from $p$ to which some, suitably parametrized, subsequence of $\sigma_n$ converges uniformly on compact subset. But then $\sigma^p$ would contradict non-total imprisonment.
\end{proof}

In a globally hyperbolic spacetime for every compact subset $K$, $J^+(K)\cap J^{-}(K)$ is compact.
As a consequence, global hyperbolicity trivially implies causal pseudoconvexity (and hence maximal causal pseudoconvexity). However, causal pseudoconvexity does not imply global hyperbolicity, e.g.\ a strip $\vert x\vert<1$ in Minkowski 1+1 spacetime \cite{BEE:96}.

We now give a simplified proof of the next Lorentzian result originally proved in  \cite[Th.\ 1]{beem92}. In the last section we shall show that the reverse implication does not hold.

\begin{theorem}\label{thm:beem92}
If $(M,g)$ is causally simple, then it is maximally null pseudoconvex (equivalently, it satisfies the  maximal null LGS).
\end{theorem}

The equivalence in the different formulations of the conclusion  is given by Lemma \ref{lem:mnpc}, since causal simplicity implies non-total imprisonment.

\begin{proof}
Suppose that the claim is false, then we can find a compact set $K$, and maximal null geodesic segments $\gamma_n$ with endpoints $p_n,q_n\in K$, such that for some $r_n\in \gamma_n$, $r_n$ escapes every compact set. Further we can assume $p_n\to p\in K$, and $q_n\to q\in K$. As $(p_n,q_n)\in J$ we have by causal simplicity $(p,q)\in J$. Moreover, it cannot be $(p,q)\in I$ otherwise for sufficiently large $n$, $(p_n,q_n)\in I$ which contradicts the maximality of $\gamma_n$. Thus $(p,q)\in E=J\setminus I$.

By the limit curve theorem \cite[Th.\ 2.53(ii)]{minguzzi18b} there are lightlike rays $\sigma^p$ starting from $p$ and $\sigma^q$ ending at $q$ such that for every $p'\in \sigma^p$ and $q'\in \sigma^q$, we have $(p',q')\in \bar J=J$ where we used causal simplicity. For every $p'\in \sigma^p\backslash\{p\}$ we have $(p',q)\in J$ but the causal curve connecting $p'$ to $q$ must be the prolongation of the lightlike ray $\sigma^p$ otherwise $(p,q)\in I$, a contradiction. Thus $q\in \sigma^p$ and repeating the argument by taking $p'$ along $\sigma^p$ after $q$ one gets that $\sigma^p$ passes through $q$ several times, which violates causality.
\end{proof}

Next we establish that for Riemannian manifolds convexity is stronger than the minimal LGS. Again, in the last section we shall show that the reverse implication does not hold.

\begin{theorem} \label{man}
If $\Sigma$ is convex, then it  satisfies the minimal LGS (equivalently, it is minimal pseudoconvex).
\end{theorem}

The equivalence in the different formulations of the conclusion  is given by Lemma \ref{lem:mpc}.

\begin{proof}
Let $\sigma_n: [0,a_n]\to \Sigma$ be minimizing unit speed geodesics such that $p_n:=\sigma_n(0)\to p$ and $q_n:=\sigma_n(a_n)\to q\ne p$, $p,q\in \Sigma$. Without loss of generality we can assume that $\dot \sigma_n(0)\to u \in T_p \Sigma$ and that $a_n \to a>0$. Let $\sigma:I\to \Sigma$ be the unit speed geodesic that starts from $p$ with tangent $\dot \sigma(0)=u$. For every $t\in I$ by the continuity of the exponential map
\[
\sigma_n(t)=\exp_{p_n}(\dot \sigma_n(0) t)\to  \exp_{p}(u t)=\sigma(t).
\]
Moreover for every $c\in I$, $\sigma_n\vert_{[0,c]}$ is minimizing and so $\sigma\vert_{[0,c]}$ is minimizing. If there is $t\in I$ such that $\sigma(t)=q$, then it must be $t\ge a$ otherwise, as $t$ is the length of a curve (the curve $\sigma$) connecting $p$ to $q$, no $\sigma_n$ could be minimizing for sufficiently large $n$. However, $t>a$ cannot happen since otherwise it would be shorter to go from $p$ to $q$, passing from some $\sigma_n$, which would contradict that $\sigma\vert_{[0,t]}$ is minimizing. We conclude that if there is $t$ such that $\sigma(t)=q$ then $t=a$.

We can now assume that $q\notin \sigma$ otherwise we have finished, hence we can assume $I=[0,b)$. Observe that $d(p_n,q_n)=a_n$ and by the continuity of distance $d(p,q)=a$.
Now, for $0<\epsilon<b$,  we have $d(\sigma_n(\epsilon), q_n)=a_n-\epsilon$ as $\sigma_n$ is minimizing, and hence $d(\sigma(\epsilon), q)=\lim_n  d(\sigma_n(\epsilon), q_n)=a-\epsilon$ as $d$ is continuous. Thus any chosen minimizing curve $\gamma$ connecting $\sigma(\epsilon)$ with $q$ has length $a-\epsilon$. But  $d(p,q)<\ell(\sigma \vert_{[0,\epsilon]})+\ell(\gamma)\le \epsilon+a-\epsilon=a$ where the first inequality is strict because there must be a corner at $\sigma(\epsilon)$ between $\sigma$ and $\gamma$, otherwise it would be $q\in \sigma$. The contradiction proves that the connecting case is the only option.
\end{proof}

We do not known if pseudoconvexity  implies convexity. The next result represents an attempt in this direction. It is not used in what follows.

\begin{proposition}
Let $(\Sigma,h)$ be a Riemannian manifold which admits an equidimensional embedding into a complete manifold
$(\tilde{\Sigma},\tilde{h})$. Then $(\Sigma,h)$ is minimally pseudoconvex if and only if it is convex.
\end{proposition}

\begin{remark}
The proof shows that a positive lower bound on the convexity radius on bounded sets of $(\Sigma,h)$ suffices at least for the conclusion that
minimal pseudoconvexity implies convexity.
\end{remark}

\begin{proof}
(i) By Theorem \ref{man} convexity implies minimal pseudoconvexity without further assumptions.

(ii) Assume that $(\Sigma,h)$ is minimally pseudoconvex. We will show that the closure $\overline{\Sigma}$ of $\Sigma$ in $\tilde{\Sigma}$ is locally
convex, i.e. for any $x\in \overline{\Sigma}$ there exists $\epsilon>0$ such that $\overline{\Sigma}\cap B^{\tilde{h}}_\epsilon(x)$ is convex in $\tilde{\Sigma}$ (i.e.\ any two points in this set are connected by a geodesic segment which is minimizing among all the connecting curves contained in the same set).

We define the convexity radius for $(\tilde \Sigma, \tilde h)$ as in  \cite[Corollary 1.9.11]{K95}. All we need to know is that this function $r\colon \tilde{\Sigma}\to (0,\infty]$ is continuous and that in $B^{\tilde{h}}_{r(\tilde{x})}(\tilde{x})$  any two points  are connected by one and only one geodesic contained in the  ball, and this geodesic is  minimizing among all the connecting curves in $\tilde \Sigma$.

We claim that for all $\tilde x\in \tilde \Sigma$ and all $y,z$ in the same connected component of $B^{\tilde h}_{r(\tilde x)}(\tilde x)\cap \Sigma $ the unique minimal $\tilde{h}$-geodesic between $y$ and $z$ lies
in $\Sigma$.
This can be seen as follows. Let $\eta\colon [0,1]\to B^{\tilde h}_{r(\tilde x)}(\tilde x)\cap \Sigma$ be a curve from $y$ to $z$.
The set of parameters $t\in [0,1]$ for which the unique minimal $\tilde{h}$-geodesic between $y$ and $\eta(t)$ lies in $\Sigma$ is nonempty.
It is also open, since $\Sigma\subset \tilde{\Sigma}$ is open and $\exp^{\tilde{h}}_y$ is a  diffeomorphism defined on $\exp^{\tilde{h}}_y{}^{-1}[B^{\tilde h}_{r(\tilde x)}(\tilde x)]$. Finally it is closed by minimal pseudoconvexity.
Therefore the unique minimal
$\tilde{h}$-geodesic between $y$ and  $\eta(1)=z$ lies in $\Sigma$. This geodesic is also a minimal geodesic for $(\Sigma, h)$ as $\text{dist}^{\tilde{h}}\le \text{dist}^{{h}}$, and hence the unique minimal geodesic for $(\Sigma, h)$ connecting $y$ and $z$. By considering limits of geodesics we infer that,
  for every $x,y\in B^{\tilde h}_{r(\tilde x)}(\tilde x)$ belonging to the closure of the same connected component $C$ of  $  B^{\tilde h}_{r(\tilde x)}(\tilde x) \cap \Sigma $ there is a $\tilde h$-minimizing geodesic entirely contained in $\overline{C}\subset \overline \Sigma$ connecting them.

Observe that for every point $\overline{x}\in \overline{\Sigma}$ we have that every connected component of $B^{\tilde{h}}_{r(\overline{x})}(\overline{x})\cap
\Sigma$ is convex, and that the established properties provide a ``local convexity'' property for $\overline{\Sigma}$.

Now let $x,y\in\Sigma$ be given. We will show that there exists a minimal geodesic in $\Sigma$ connecting the two points. 

The following reasoning is similar to \cite[Satz 2.8(i)]{Ba78}. Let $\gamma_n\colon [0,1]\to \Sigma$ be a minimizing sequence of curves connecting $x$ with 
$y$. According to the argument above we can assume that the curves are geodesic polygons and that the length of the individual arcs is bounded from below by some 
positive constant. By the completeness of $(\tilde{\Sigma},\tilde{h})$ a subsequence converges to a geodesic polygon $\gamma\colon [0,1]\to \overline{\Sigma}$. Using
a standard argument involving the triangle inequality we see that $\gamma$ is a $\tilde{h}$-geodesic.

We claim that $\gamma$ does not intersect the boundary $\partial\Sigma:=\overline{\Sigma}\setminus \Sigma$. Assume otherwise, i.e. there exists $t\in (0,1)$ with $\gamma(t)\in \partial\Sigma$. W.l.o.g. we can assume that $t$ is minimal in $(0,1)$.
Denote with $\Sigma_t$ the connected component of $B^{\tilde{h}}_{r(\gamma(t))/4}(\gamma(t))\cap \Sigma$ which contains a segment $\gamma|_{[s,t)}$
for some $s<t$. Note that $\Sigma_t$ is convex with $\gamma(t)\in \partial \Sigma_t$. The following argument is an adaptation of \cite[Lemma 8.6]{CE75}
to the present situation. Choose an arbitrary small smooth
hypersurface $W\subset \Sigma_t$ transversal to $\gamma$ and containing $\gamma(s)$. Choose $u>t$ such that $\gamma|_{[t,u]}\subset
B^{\tilde{h}}_{r(\gamma(t))/4}(\gamma(t))$ and consider the set
\[
V:=\{\exp_{r}(\lambda w)|\; \tilde{h}(w,w)<r(\gamma(t))^2,\; \exp_{r}(w)\in W,\; \lambda \in (0,1)\}.
\]
Here $r\in \Sigma$ is a point sufficiently close to $\gamma(u)\in \overline{\Sigma}$ which stays in the connected component $\Sigma_t$ (it is possible to show that it exists perturbing $\gamma_{[s,u]}$).
The set $V$ is open and contained in ${\Sigma_t}$ by convexity. Further $V$ is an open neighborhood of $\gamma(t)$, (remember that $\gamma$ is a $\tilde h$-geodesic hence $C^1$) a contradiction to the
assumption that $\gamma(t)\in \partial\Sigma_t$. Therefore $\gamma$
is contained in $\Sigma$.
 The lower semicontinuity of the length functional implies that
the curve is minimal in $\Sigma$. Note that $\gamma$ is i.g. not minimal in $\tilde{\Sigma}$.
\end{proof}

\section{Pseudoconvexity in static spacetimes}\label{sec:stat}

Let $(\Sigma,h)$ be a Riemannian manifold. In this section we  investigate the relation of convexity and pseudoconvexity properties of $(\Sigma, h)$ to causality and pseudoconvexity properties of the static spacetime $(M,g)$ with
\begin{equation}\label{eq:static}
 M:=\mathbb{R} \times \Sigma, \quad\text{and}\quad g=-dt^2+h.
\end{equation}
The projection onto the first factor $t:M\to \mathbb{R}$  is a temporal function.


Observe that in semi-Riemannian product manifolds a path is, up to para\-me\-tri\-sation, a geodesic if and only if the projections onto the factors are geodesics. Hence in our case, geodesics are always of the form $\gamma=(\alpha, \sigma)$ where $\alpha$ is some linear function in $\mathbb{R}$ and $\beta$ is a geodesic in $\Sigma$. Moreover, clearly $\gamma$ is causal if and only if $|\alpha'| \geq \|\sigma'\|_h$.
We observe the following relation of minimizing and maximizing geodesics.
\begin{lemma}\label{lem:geodesiceq}
	A causal geodesic $\gamma= (\alpha, \sigma)$ in $M$ is maximizing if and only if $\sigma$ is minimizing in $\Sigma$.
\end{lemma}

\begin{proof}
It suffices to consider geodesics $\gamma=(\alpha,\sigma)$ of the form $\alpha:[0,1]\to{\mathbb R}$, $\alpha(s)=bt$, and $\sigma:[0,1]\to\Sigma$, i.e., $\|\sigma'(s)\|_h=L^h(\sigma)$, connecting $p=(0,x)$ and $q=(b,y)$ with $b>0$. First observe that $\gamma$ is causal iff $b\geq L^h(\sigma)$ and so
\begin{equation}\label{eq:Lg}
 L^g(\gamma)=\sqrt{b^2-L^h(\sigma)^2}.
\end{equation}

Now suppose $\sigma$ is not minimizing, then there is $\tilde\sigma:[0,1]\to\Sigma$ connecting $x$ to $y$ with $L^h(\tilde\sigma)<L^h(\sigma)$. But then for $\tilde\gamma=(bt,\tilde\sigma)$ we find
$L^g(\tilde\gamma)=\sqrt{b^2-L^h(\tilde\sigma)^2}>\sqrt{b^2-L^h(\sigma)^2}= L^g(\gamma)$ and so $\gamma$ is not maximizing.

Conversely, suppose that $\gamma$ is not maximizing, then there is a future pointing timelike curve $\tilde \gamma:[0,1]\to M$, $\tilde\gamma=(bt,\tilde\sigma)$ with $L^g(\tilde\gamma)>L^g(\gamma)$.
Since $\tilde \gamma$ need not be a geodesic we can in general not parametrize it both linear in the first factor and constant speed in the second. However, we have
 \[
 L^g(\tilde{\gamma})=\int_0^{1} \sqrt{b^2-\|\tilde{\sigma}'\|_h^2}dt>L^g(\gamma)=\sqrt{b^2-(L^h(\sigma))^2}.
 \]
By applying  the Cauchy-Schwarz inequality to $\int_0^{1} \sqrt{b^2-\|\tilde{\sigma}'\|_h^2}dt$ and by using the previous inequality we get the  second inequality in the next expression
 \[
 L^h(\tilde{\sigma})^2\le \int_0^{1} \|\tilde{\sigma}'\|_h^2dt<L^h(\sigma)^2.
 \]
The first inequality is obtained through another standard application of the Cauchy-Schwarz inequality. In conclusion,  $\sigma$ is not minimizing between $x$ and $y$.
\end{proof}

\begin{lemma}\label{lem:max-min-pc}
If $M$ is maximally null pseudoconvex, then  $\Sigma$ is minimal pseudoconvex. Conversely, if $\Sigma$ is minimal pseudoconvex, then $M$ is maximally causal (hence null) pseudoconvex.
$M$ is pseudoconvex if and only if $\Sigma$ is pseudoconvex.
\end{lemma}

\begin{proof}
Let $\Sigma$ be pseudoconvex and take any $K\subseteq M$. Then the projections of $K$ onto $\mathbb{R}$ and $\Sigma$ are compact, i.e.\ we have $K \subseteq [c,d] \times C$ for $c,d\in \mathbb{R}$ and some compact $C\subseteq \Sigma$. By pseudoconvexity of $\Sigma$ there exists
a compact set $C'$ such that any geodesic starting and ending in $C$ is contained in $C'$. Now let $\gamma$ be a geodesic in $M$, starting and ending in $K$. We can write $\gamma(t)=(a+bt, \sigma(t))$, hence $c\leq a, a+b\leq d$ with $\sigma$ a geodesic starting and ending in $C$. This however implies that $\gamma \subseteq K':= [c,d]\times C'$.

The proof that ``$\Sigma$ is minimally pseudoconvex'' implies ``$M$ is maximally causal pseudoconvex'' follows from a very similar argument which makes use  of Lemma \ref{lem:geodesiceq}.
\medskip

For the converse direction, we first look at the maximal-minimal case:
Let $M$ be maximally null pseudoconvex and let $C\subseteq\Sigma$ be compact. Any minimizing unit speed geodesic $\sigma$ starting and ending in $C$ fulfills $L^h(\sigma) \leq c:=\text{diam}^h(C)$. For the compact set $K:=[0,c]\times C$ by maximal null pseudoconvexity there exists some compact set $K'$ containing all maximal null geodesics starting and ending in $K$. Again by construction we must have $K' \subseteq [0,c] \times C'$ for some compact set $C'$ in $\Sigma$.
By Lemma \ref{lem:geodesiceq} the null geodesic $\gamma(t):=(t,\sigma(t))$ is maximizing and hence contained in $K'$, but then also $\sigma \subseteq C'$.

If $M$ is pseudoconvex then $\Sigma$ is too by noticing that for any compact set $C\subseteq \Sigma$ the set $\{0\} \times C$ is compact in $M$.
\end{proof}

\begin{remark}
While pseudoconvexity of $\Sigma$ clearly implies causal and null pseudoconvexity of $M$ the converse does not hold. This is due to the fact that any causal geodesic starting and ending in a compact set in $M$ has a projection with a  length bounded by the difference of the $t$ components of the causal curve endpoints. So one cannot expect to gain control over sequences in $\Sigma$ with unbounded lengths.
\end{remark}

\begin{lemma}\label{lem:conv}
 $\Sigma$ is convex if and only if $M$ is causally simple.
\end{lemma}
\begin{proof}

Assume $\Sigma$ is convex. Observe that $(M,g)$ is causal as $t$ is a time function. By translational invariance over the time fiber and by time reflection symmetry, we need only to show that  $J^+(p)$ is closed for any point of the form $p=(0,x)$.

Let $q_n=(b_n,y_n) \in J^+(p)$ with $q_n=(b_n, y_n) \to (b,y)=q$. There exist causal paths $\sigma_n$ from $p$ to $q_n$ of the form $(\frac{b_n}{d^h(x,y_n)}\,t,\sigma_n(t))$ with $\sigma_n$ from $x$ to $y_n$.
The causality condition reads $\|\sigma_n'\|_h \le b_n/d^h(x,y_n)$.
Hence $ d^h(x,y_n) \leq L^h(\sigma_n) = \int_0^{d^h(x,y_n)} \|\sigma_n'\|_h \leq b_n$, which by continuity gives $d^h(x,y) \leq b$.

By convexity there exists a unit speed geodesic $\sigma:[0,d^h(x,y)] \to \Sigma$ connecting $x$ to $y$. Now $\gamma:[0,d^h(x,y)] \to M, \gamma(t):=(\frac{b}{d^h(x,y)}\, t, \sigma(t))$ is a path from $p$ to $q$ that is causal iff, $d^h(x,y) \leq b$, which is the case as we just proved. Thus $q\in J^+(p)$ and, by the arbitrariness of $q$, $J^+(p)$ is closed.
\medskip

Conversely, let $M$ be causally simple and let $x,y\in \Sigma$. There exist curves $\sigma_n$ from $x$ to $y$ with $L^h(\sigma_n)=:l_n \searrow l:=d^h(x,y)$ with $\|\sigma_n'\|_h=1$. The curves $\gamma_n:[0,l_n]\to M, \gamma_n(t):=(t,\sigma_n(t))$ from $(0,x)$ to $(l_n,y)$ are null and hence $(l,y)\in J^+((0,x))$ by causal simplicity. So there exists some causal curve $\alpha:[0,l]\to M, \alpha(t)=(t,\beta(t))$ with $\beta$ a path in $\Sigma$ from $x$ to $y$.
	Moreover $0 \geq g(\alpha', \alpha')= -1 + \|\beta'\|_h^2$ and so $ \|\beta'\|_h^2\leq 1$, which leads to $L^h(\beta) \leq l$ and hence $\beta$ must be a minimizing geodesic from $x$ to $y$.
\end{proof}

\section{The counterexamples}\label{sec:r-counter}

With Thm.\ \ref{man} we learned that on a Riemannian manifold convexity implies the minimal LGS.
In this section we show that the converse does not hold.

Then we shall consider the Lorentzian Thm.\ \ref{thm:beem92}, analogous to Thm.\ \ref{man}, according to which causal simplicity implies the maximal null LGS. By using the results of the previous section on static spacetimes, we shall show, once again, that the converse implication does not hold.

Let us give an example of Riemannian space $(\Sigma,h)$ which satisfies the minimal LGS but is not convex.
A very similar Riemannian space was constructed in \cite[Sec.\ 2.1]{bartolo02b} as a counterexample to another statement also related to the notions of convexity and connectedness in Riemannian spaces.

\begin{figure}[ht]
\centering
\includegraphics[width=.75\textwidth]{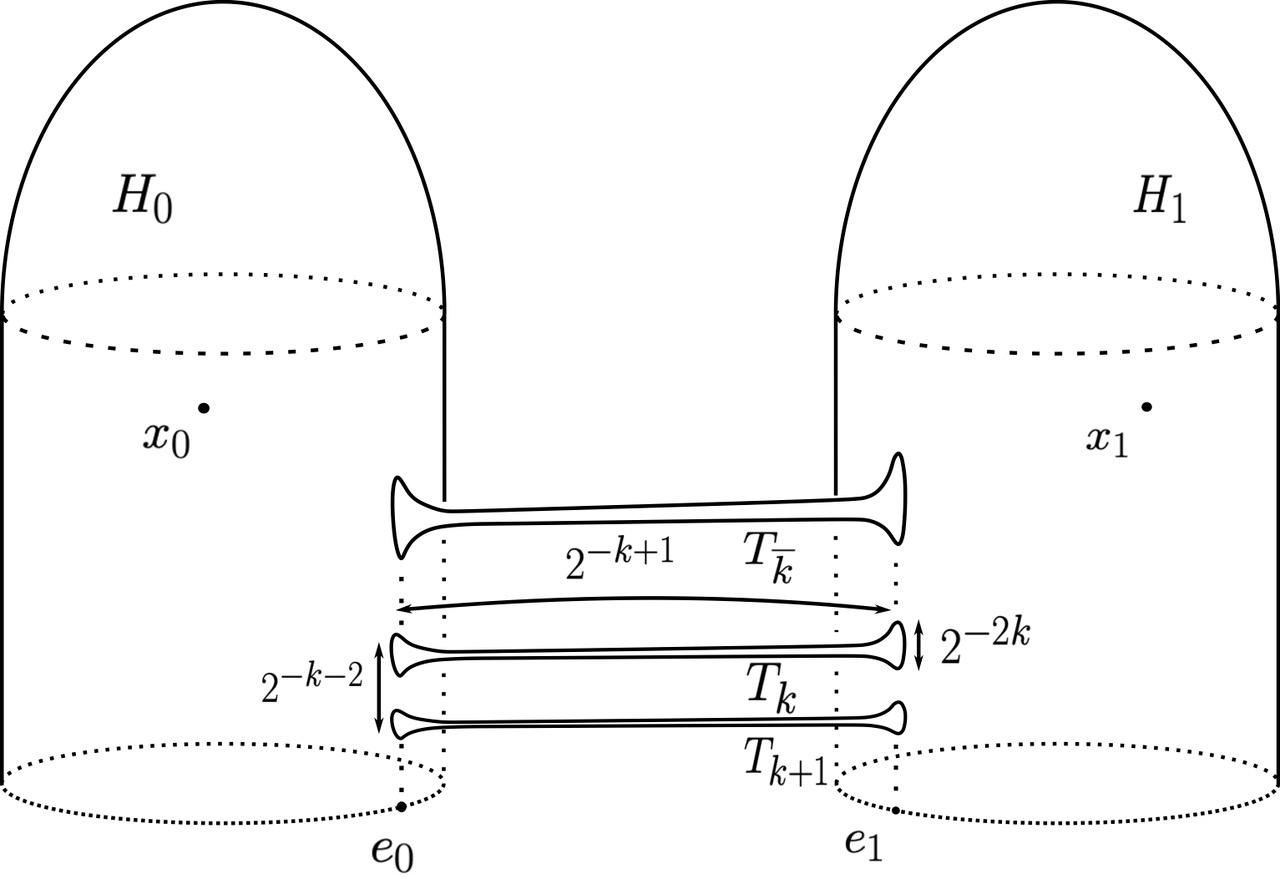}
\caption{The counterexample consists of two cylinders $H_0$ and $H_1$, closed above with cups and connected by a sequence of immersed tubes $T_k$, $k \ge \bar k$. This figure is inspired by that in  \cite[Sec.\ 2.1]{bartolo02b}.}
\end{figure}

\begin{example}[A non-convex Riemannian manifold possessing the minimal LGS]\label{ex:r}
 The space $\Sigma$ consists of two  cylinders $H_0$ and $H_1$ closed above with cups and then  connected in the flat region by a sequence of immersed tubes $T_k$, $k\ge \bar k>0$.  The Cauchy boundary of the space can be identified with the union of two circles, which are the lower boundaries of the cylinders  used in the construction.
The mouths of the tubes converge to points $e_0$ and $e_1$ on the Cauchy boundary of the space.

Let the geodesic distance between the center of the mouths of $T_k$ and $T_{k+1}$ in  $H_0$ before the  discs (tube mouths) are excised,  be $\frac{1}{2^{k+2}}$, let the diameter of the mouths (discs) in the same geometry be $\frac{1}{4^{k}}$ (basically is goes so fast to zero that in our arguments this distance becomes  negligible),   and let the length of the tube $T_k$ from mouth to mouth in $(\Sigma, h)$ be $\frac{1}{2^{k-1}}$. Let analogous conditions hold in $H_1$. If $\gamma_1$ is a curve passing through $T_k$ from mouth to mouth (here $\gamma_1$ might have endpoints belonging to some tubes that might coincide with $T_k$ or not) then we can find a curve $\gamma_2$ connecting the same endpoints of $\gamma_1$ passing through $T_{k+1}$ from mouth to mouth and such that
\[
\ell(\gamma_2)\le  \ell(\gamma_1) +2 \frac{1}{2^{k+2}}- \frac{1}{2^{k-1}}   +\frac{1}{2^{k}} +O(\frac{1}{4^k}) \le \ell(\gamma_1) -\frac{1}{2^{k+1}}  +O(\frac{1}{4^k}) .
\]
This implies  that, provided the space is defined with  $\bar k$ sufficiently large,  $\ell(\gamma_2) <\ell(\gamma_1)$, hence
 a minimizing geodesic cannot pass through a tube from mouth to mouth and so that there are certainly pairs of points not connected by a minimizing geodesic, for instance any pair $(x_0, x_1)$ where $x_0$ belongs to the first cylinder and $x_1$ belongs to the second cylinder. As a consequence, convexity does not hold (this fact was already pointed out in \cite{bartolo02b}).

We want to show that the minimal LGS holds.
Let $\sigma_n\colon [0,a_n]\to \Sigma$ be minimizing  unit speed geodesics such that $p_n:=\sigma_n(0)\to p$ and $q_n:=\sigma_n(a_n)\to q$, $p,q\in \Sigma$. Without loss of generality we can assume that the tangents $\dot \sigma_n$ converge to some unit vector $u\in T_p \Sigma$, and since the diameter of  $(\Sigma, h)$ is bounded we can also assume that $a_n\to a$.
Let $\sigma\colon I\to \Sigma$ be the geodesic  that starts from $p$ with tangent $u$. If its domain interval includes $[0,a]$,
then by the continuity of the exponential map we can conclude that
\[
q_n=\exp_{p_n} ( \dot \sigma_n(0) a_n) \to \exp_{p} ( u a) = \sigma(a)
\]
that is $q=\sigma(a)$. Observe that $a_n$ is the length of $\sigma_n$ and it converges to the length $a$ of $\sigma$. Thus $\sigma$ must be minimizing otherwise for sufficiently large $n$, $\sigma_n$ would not be minimizing.
This would give the desired
result so we have only to show that the domain of $\sigma$ includes $[0,a]$.

Suppose not. The maximal domain of $\sigma$ is $[0,b)$ for some $b\le a$, and by standard ODE theory \cite{hartman02} $\sigma(t)$ escapes every compact set as $t\to a$ (though possibly returning indefinitely to it).

Let $T^p=T_s$ be the tube whose mouth is closest to $p$ (which could be the tube to which $p$ belongs), and let $T^q=T_t$ be the tube whose mouth is closest to $q$ (if there are more choices for the closest tube we choose that with larger $s$, resp.\ $t$). For sufficiently large $k$ namely for  $k>K> \max(s,t)$, with suitable $K>0$ the minimizing geodesics $\sigma_n$ cannot transverse $T_k$  from mouth to mouth (as they are minimizing, cf.\ the above argument).

Let us section the figure by cutting with a horizontal plane at a height $y>0$ selected so that above it we have  the points $p$, $q$, and the tubes $T_k$, $k\le K$, thus including $T^p$ and $T^q$. Then the geodesics $\sigma_n$  stay entirely above the horizontal section (because if they cross the plane then they can be replaced by a shorter curve running on the plane for a geodesic segment, a fact which would contradict the minimizing property). As a consequence $\sigma$ being a limit of $\sigma_n$ is also entirely above the plane and so it cannot escape every compact set of $\Sigma$. The contradiction proves that the domain of $\sigma$ is $[0,a]$ and so that the minimal LGS holds.
\end{example}

We are ready to give the Lorentzian example.

\begin{example}[A maximally causal pseudoconvex and causally continuous but non-causally simple spacetime] We consider the static spacetime $(M,g)$ with $M={\mathbb R}\times\Sigma$ and $g=-dt^2+h$, where $(\Sigma,h)$ is as in Example \ref{ex:r}.
By construction, $M$ is strongly causal (actually stably causal, because of the time function $t$). Due to the timelike Killing vector $\p_t$ it is reflective and hence causally continuous \cite{minguzzi18b}.

Now, since $\Sigma$ possesses the minimal LGS, it is also minimal pseudoconvex by Lemma \ref{lem:mpc} and hence $M$ is maximal causal (hence null) pseudoconvex by Lemma \ref{lem:max-min-pc}.

On the other hand $\Sigma$ is not convex and so,  by Lemma \ref{lem:conv}, $M$ is not causally simple.
\end{example}

\section*{Acknowledgements}

B.S. is supported by the Uni:Docs program of the University of Vienna.
R.S. is supported by FWF-grants P28770 and P33594.
E.M. is partially supported by GNFM of INDAM.
J.H. and S.S. are partially supported by the SFB/TRR 191 ``Symplectic Structures in Geometry, Algebra and Dynamics'', funded by the Deutsche Forschungsgemeinschaft.
Part of this work has been developed while S.S. was a visiting professor at Politecnico di Bari; he thanks Politecnico di Bari and the Department of Mechanics, Mathematics and Management for the support.

\end{document}